\documentclass[11pt]{amsart}
\usepackage[ruled,vlined,linesnumbered,norelsize]{algorithm2e}
\usepackage{booktabs}
\usepackage{url}

\newtheorem{thm}{Theorem}

\newcommand{\ZZ}{\mathbf{Z}}
\newcommand{\FF}{\mathbf{F}}

\newcommand{\assign}{\leftarrow}
\newcommand{\sfloor}[1]{\lfloor {#1} \rfloor}

\begin{document}

\title{Faster arithmetic for number-theoretic transforms}
\author{David Harvey}

\begin{abstract}
We show how to improve the efficiency of the computation of fast Fourier transforms over $\FF_p$ where $p$ is a word-sized prime. Our main technique is optimisation of the basic arithmetic, in effect decreasing the total number of reductions modulo $p$, by making use of a redundant representation for integers modulo $p$. We give performance results showing a significant improvement over Shoup's NTL library.
\end{abstract}

\maketitle

\section{Introduction}

An important problem in computational number theory and cryptography is the efficient implementation of modular arithmetic. A typical implementation strategy is to represent elements of $\ZZ/N\ZZ$ by residues in a standard interval, such as $[0, N)$ or $[-N/2, N/2)$, and to reduce intermediate results to this interval after each operation in $\ZZ/N\ZZ$, such as addition or multiplication.

In many algorithms one can obtain a substantial constant factor speedup by delaying the modular reduction until after several arithmetic operations have been performed, by taking into account the bit-size of intermediate results. For example, to compute a dot product $\sum_i a_i b_i$, a fundamental operation in linear algebra, one may accumulate terms in batches, using ordinary integer (or even floating-point) arithmetic, and perform the reduction modulo $N$ after each batch \cite{DGP-linalg}.

The aim of this paper is to describe a strategy for reducing the number of modular reductions in the computation of a discrete Fourier transform over a finite field, also known as a \emph{number-theoretic transform} (NTT). The NTT has a vast range of applications; we mention here only fast multiplication of large integers or polynomials \cite{vzGG-compalg}.

For simplicity we restrict to the following situation. Let $\beta$ describe the machine word size, for example $\beta = 2^{32}$ or $\beta = 2^{64}$. We work over $\FF_p = \ZZ/p\ZZ$ where $p$ is a word-sized prime, and we assume that $p = 1 \pmod L$, where $L = 2^\ell$ is the transform length. This ensures that $\FF_p$ contains a primitive $L$-th root of unity; we fix one of these, denoted $\omega$. The NTT is the map $(\FF_p)^L \to (\FF_p)^L$ defined by
 \[ b_j = \sum_{i=0}^{L-1} \omega^{ij} a_i, \qquad 0 \leq j < L. \]
Equivalently, this is the map that evaluates the polynomial $a_0 + a_1 x + \cdots + a_{L-1} x^{L-1}$ at the points $1, \omega, \omega^2, \ldots, \omega^{L-1}$.

It is well known that the fast Fourier transform (FFT) can be used to compute the NTT using $O(L \log L)$ operations in $\FF_p$. For completeness, a simple in-place iterative radix-2 FFT algorithm is shown in Algorithm \ref{algo:fft}.
\begin{algorithm}
\label{algo:fft}
\SetAlgoLined
\DontPrintSemicolon
\KwIn{$\omega \in \FF_p$ and $x = (x_0, \ldots, x_{L-1}) \in (\FF_p)^L$, with $L = 2^\ell$}
\KwOut{DFT of $x$ with respect to $\omega$, in bit-reversed order}
\For{$i \assign 1, 2, \ldots, \ell$}{
  $\zeta \assign \omega^{2^{i-1}}$\;
  $m \assign 2^{\ell-i}$\;
  \For{$0 \leq j < 2^{i-1}$}{
    $t \assign 2jm$\;
    \For{$0 \leq k < m$}{
      $\begin{bmatrix} x_{t+k} \\ x_{t+k+m} \end{bmatrix} \assign
           \begin{bmatrix} x_{t+k} + x_{t+k+m} \\ \zeta^k (x_{t+k} - x_{t+k+m}) \end{bmatrix}$ \nllabel{line:butterfly}\;
    }
  }
}
\caption{Simple FFT}
\end{algorithm}

Our main focus in this paper is on the \emph{butterfly operation} in line \ref{line:butterfly}, which computes $(X, Y) \mapsto (X + Y, W(X - Y))$ for some fixed root of unity $W \in \FF_p$ (a suitable power of $\omega$). At first glance this requires one modular addition, one modular subtraction, and one modular multiplication, and this is how the butterfly is usually implemented. Of course these modular operations are themselves implemented on modern microprocessors using more basic primitives. For example, a modular addition is usually implemented as an ordinary integer addition, followed by a comparison with $p$, followed by a conditional subtraction. In this paper we investigate the butterfly at this lower level, showing how to streamline the implementation to reduce the number of comparisons and conditional operations. Another interpretation is that we have reduced the number of modular reductions.

\section{A typical butterfly implementation}

Victor Shoup's NTL (Number Theory Library) \cite{ntl-5.5.2} is a popular C++ library used in computational number theory. It makes heavy use of the NTT. Its implementation of the butterfly $(X, Y) \mapsto (X + Y, W(X - Y))$ may be expressed by the pseudocode shown in Algorithm \ref{algo:ntl}. It has been simplified to focus attention on the arithmetic aspects, ignoring issues like loop unrolling, software pipelining, and locality. All variables represent register-sized quantities.
\begin{algorithm}
\label{algo:ntl}
\SetAlgoLined
\DontPrintSemicolon
\KwIn{$p < \beta/2$ \newline $0 < W < p$ \newline $W' = \sfloor{W\beta/p}$, $0 < W' < \beta$ \newline $0 \leq X < p$ \newline $0 \leq Y < p$}
\KwOut{$X' = X + Y \pmod p$ \newline $Y' = W(X - Y) \pmod p$}
\medskip
$X' \assign X + Y$ \nllabel{line:ntl-add-1}\;
\lIf{$X' \geq p$}{$X' \assign X' - p$} \nllabel{line:ntl-add-2}\;
\medskip
$T \assign X - Y$ \nllabel{line:ntl-sub-1}\;
\lIf{$T < 0$}{$T \assign T + p$} \nllabel{line:ntl-sub-2}\;
\medskip
$Q \assign \lfloor W'T/\beta \rfloor$ \nllabel{line:ntl-mul-1}\;
$Y' \assign (WT - Qp) \bmod \beta$ \nllabel{line:ntl-mul-2}\;
\lIf{$Y' \geq p$}{$Y' \assign Y' - p$} \nllabel{line:ntl-mul-3}\;
\medskip
\KwRet{$X', Y'$}\;
\caption{NTL butterfly implementation}
\end{algorithm}

Note that $W$ and $W'$ are independent of the data being transformed; they can be precomputed and reused for each transform. A similar comment applies to all the butterfly algorithms considered in this paper.

\begin{thm}
\label{thm:ntl}
Algorithm \ref{algo:ntl} is correct.
\end{thm}
\begin{proof}
Lines \ref{line:ntl-add-1}--\ref{line:ntl-add-2} compute the sum $X + Y$ and reduce it modulo $p$ to the standard interval $[0, p)$, using the assumption $p < \beta/2$ to avoid overflow in the first line. Lines \ref{line:ntl-sub-1}--\ref{line:ntl-sub-2} compute $T = X - Y \pmod p$ in the same way, assuming that $T$ has a signed type for the comparison.

Lines \ref{line:ntl-mul-1}--\ref{line:ntl-mul-3} compute the product $WT \pmod p$. Line \ref{line:ntl-mul-1} first generates an estimated quotient $Q$. By the definition of $W'$ and $Q$ we have
 \[ 0 \leq \frac{W\beta}p - W' < 1, \qquad 0 \leq \frac{W'T}{\beta} - Q < 1. \]
Multiplying by $Tp/\beta$ and $p$ respectively, and adding, yields
 \[ 0 \leq WT - Qp < \frac{Tp}{\beta} + p < 2p < \beta. \]
In particular, line \ref{line:ntl-mul-2} correctly computes $Y' = WT - Qp$, and the single correction in line \ref{line:ntl-mul-3} suffices to reduce it into $[0, p)$.
\end{proof}

The operations performed by Algorithm \ref{algo:ntl} map comfortably onto modern instruction sets. In line \ref{line:ntl-mul-1}, the expression $\lfloor W'T/\beta \rfloor$ represents the high word of the product $W'T$, which can typically be obtained by a single hardware multiply instruction. Both $WT$ and $Qp$ in line \ref{line:ntl-mul-2} require only the low word of the product. The conditional additions and subtractions in lines \ref{line:ntl-add-2}, \ref{line:ntl-sub-2} and \ref{line:ntl-mul-3} are simple enough that a modern compiler will implement them by a conditional move instruction rather than by a branch.

If we assume that $X$ and $Y$ are distributed uniformly in $[0, p)$, then the first two conditions (lines \ref{line:ntl-add-2} and \ref{line:ntl-sub-2}) hold with probability 50\%. The behaviour of the third condition (line \ref{line:ntl-mul-3}) is more complex. If the quantity $W\beta/p - W'$ considered in the proof of the theorem happens to be close to zero, then the condition will be satisfied with very low probability. At the other extreme, if $W\beta/p - W'$ is close to unity, and if we assume that $p$ is near $\beta/2$, then one can show the condition is satisfied with probability about 25\%. In this case it is still reasonable to prefer a conditional move instruction to a branch.

The modular multiplication algorithm in lines \ref{line:ntl-mul-1}--\ref{line:ntl-mul-3} is attributed to Shoup in \cite{Far-thesis}, but does not seem to have been published. It first appears in NTL version 5.4 in 2005. The use of a suitable precomputed approximation to $W/p$ implies that only a single correction step (line \ref{line:ntl-mul-3}) is necessary, and that the remainder is obtained using only multiplication modulo $\beta$ (line \ref{line:ntl-mul-2}), an advantage on processors that can compute the low word faster than the full product. The latter idea is also used heavily in the division algorithm of \cite{MG-improved}.

In our exposition we assumed for simplicity that $p < \beta/2$, but Niels M\"oller has pointed out (personal communication) that this can be improved. The modular subtraction can be made to work for any $p < \beta$ by replacing the condition $T < 0$ by $X < Y$, or indeed by using the borrow generated by the subtraction $T = X - Y$. The addition can be treated similarly by rewriting it as $X' = X - (p - Y)$. A more careful analysis of the candidate remainder $WT - Qp$ then shows that the entire algorithm works for any $p < \beta/\phi$ where $\phi = \frac12(1 + \sqrt5) \approx 1.618$.

\section{A modified butterfly}

In this section we propose several modifications to Algorithm \ref{algo:ntl}. Our motivation is that the adjustment steps in lines \ref{line:ntl-add-2}, \ref{line:ntl-sub-2} and \ref{line:ntl-mul-3} are relatively expensive on modern microprocessors, compared to hardware integer multipliers, which in recent years have become very fast.

Our first observation is that Shoup's algorithm for computing $WT \pmod p$ does not require that $0 \leq T < p$; in fact the proof given above shows that it works for any $0 \leq T < \beta$. Therefore we may replace lines \ref{line:ntl-sub-1}--\ref{line:ntl-sub-2} by simply $T \assign X - Y + p$, after which we have $0 \leq T < 2p < \beta$, and the algorithm proceeds as before. This simplification is not new, although it does not appear to be well known. It is not used in NTL. It was apparently used by Fabrice Bellard in a computation of $\pi$ to 2.7 trillion decimal places in 2009 (personal communication).

Our second observation appears to be new, and is the main novelty introduced in the present paper. Namely, we may also remove the adjustment in line \ref{line:ntl-mul-3}, provided that throughout the FFT we use a \emph{redundant} representation for elements of $\FF_p$. That is, instead of representing elements of $\FF_p$ by integers in $[0, p)$, we use the wider interval $[0, 2p)$, so each element has two possible representations. For this to work, we must modify the butterfly to accept \emph{inputs} in $[0, 2p)$, and we must impose the stronger condition $p < \beta/4$. Pseudocode for the resulting butterfly is shown in Algorithm \ref{algo:new}.

According to Jason Papadopolous (personal communication), around 1998 Ernst Mayer suggested the use of a redundant representation in the context of a fast Galois transform, i.e.~a Fourier transform over $\FF_{q^2}$ where $q = 2^{61} - 1$. Our new algorithm may be regarded as a generalisation of this idea to the case of an NTT with arbitrary modulus $p$.

\begin{algorithm}
\label{algo:new}
\SetAlgoLined
\DontPrintSemicolon
\KwIn{$p < \beta/4$ \newline $0 < W < p$ \newline $W' = \sfloor{W\beta/p}$, $0 < W' < \beta$ \newline $0 \leq X < 2p$ \newline $0 \leq Y < 2p$}
\KwOut{$X' = X + Y \pmod p$, $0 \leq X' < 2p$ \newline $Y' = W(X - Y) \pmod p$, $0 \leq Y' < 2p$}
\medskip
$X' \assign X + Y$ \nllabel{line:new-add-1}\;
\lIf{$X' \geq 2p$}{$X' \assign X' - 2p$} \nllabel{line:new-add-2}\;
\medskip
$T \assign X - Y + 2p$ \nllabel{line:new-sub}\;
\medskip
$Q \assign \lfloor W'T/\beta \rfloor$ \nllabel{line:new-mul-1}\;
$Y' \assign (WT - Qp) \bmod \beta$ \nllabel{line:new-mul-2}\;
\medskip
\KwRet{$X', Y'$}\;
\caption{Modified Shoup butterfly}
\end{algorithm}

\begin{thm}
\label{thm:new}
Algorithm \ref{algo:new} is correct.
\end{thm}
\begin{proof}
Lines \ref{line:new-add-1}--\ref{line:new-add-2} compute a representative for $X + Y \pmod p$ in the interval $[0, 2p)$; this does not overflow as $p < \beta/4$. Line \ref{line:new-sub} computes a representative $T$ for $X - Y \pmod p$ in the interval $[0, 4p)$. Lines \ref{line:new-mul-1}--\ref{line:new-mul-2} compute a representative for $WT \pmod p$ in the interval $[0, 2p)$, in the same way as in the proof of Theorem \ref{thm:ntl}.
\end{proof}

Both outputs $X'$, $Y'$ lie in $[0, 2p)$, ready to be processed by a subsequent butterfly. Depending on the needs of the NTT user, it may be necessary to normalise the final NTT output into the interval $[0, p)$, imposing an additional $O(L)$ cost.

\section{Variants}

In some situations one requires a butterfly of the form $(X, Y) \mapsto (X + WY, X - WY)$. This appears if one switches from a `decimation-in-frequency' transform to a `decimation-in-time' transform. It also appears naturally in the inverse FFT obtained by running Algorithm \ref{algo:fft} backwards. Algorithm \ref{algo:inverse} shows an analogue of Algorithm \ref{algo:new} for this case. The main difference is that elements of $\FF_p$ are represented by an integer in $[0, 4p)$, so each residue has four possible representatives. As in Algorithm \ref{algo:new}, this strategy saves two modular reductions compared to the usual implementation.
\begin{algorithm}
\label{algo:inverse}
\SetAlgoLined
\DontPrintSemicolon
\KwIn{$p < \beta/4$ \newline $0 < W < p$ \newline $W' = \sfloor{W\beta/p}$, $0 < W' < \beta$ \newline $0 \leq X < 4p$ \newline $0 \leq Y < 4p$}
\KwOut{$X' = X + WY \pmod p$, $0 \leq X' < 4p$ \newline $Y' = X - WY \pmod p$, $0 \leq Y' < 4p$}
\medskip
\lIf{$X \geq 2p$}{$X \assign X - 2p$} \nllabel{line:inverse-add}\;
\medskip
$Q \assign \lfloor W'Y/\beta \rfloor$ \nllabel{line:inverse-mul-1}\;
$T \assign (WY - Qp) \bmod \beta$ \nllabel{line:inverse-mul-2}\;
\medskip
$X' \assign X + T$\;
$Y' \assign X - T + 2p$\;
\medskip
\KwRet{$X', Y'$}\;
\caption{Modified inverse butterfly}
\end{algorithm}

\begin{thm}
\label{thm:inverse}
Algorithm \ref{algo:inverse} is correct.
\end{thm}
\begin{proof}
Line \ref{line:inverse-add} reduces $X$ into $[0, 2p)$. Lines \ref{line:inverse-mul-1}--\ref{line:inverse-mul-2} compute $T = WY \pmod p$ in $[0, 2p)$, in the same way as Algorithm \ref{algo:new}. The remaining lines complete the calculation of $X'$ and $Y'$ in $[0, 4p)$.
\end{proof}

The same idea may be applied to butterflies using other modular multiplication algorithms. Algorithm \ref{algo:montgomery} shows a variant using Montgomery multiplication \cite{Mon-modular}. The idea of Montgomery multiplication is in effect to replace the usual Euclidean quotient by a $2$-adic quotient. Compared to Algorithm \ref{algo:new}, this butterfly algorithm has the advantage that only a single value $W'$ must be stored in a table for each root of unity ($W$ is not actually used in the algorithm). On the other hand, one of the multiplications modulo $\beta$ has been replaced by a full product, which may be more expensive on some processors.
\begin{algorithm}
\label{algo:montgomery}
\SetAlgoLined
\DontPrintSemicolon
\KwIn{$p$ odd, $p < \beta/4$ \newline $0 < W < p$ \newline $W' = \beta W \pmod p$, $0 < W' < p$ \newline $J = p^{-1} \pmod \beta$ \newline $0 \leq X < 2p$ \newline $0 \leq Y < 2p$}
\KwOut{$X' = X + Y \pmod p$, $0 \leq X' < 2p$ \newline $Y' = W(X - Y) \pmod p$, $0 \leq Y' < 2p$}
\medskip
$X' \assign X + Y$ \nllabel{line:mont-add-1}\;
\lIf{$X' \geq 2p$}{$X' \assign X' - 2p$} \nllabel{line:mont-add-2}\;
\medskip
$T \assign X - Y + 2p$ \nllabel{line:mont-sub}\;
\medskip
$R_1 \beta + R_0 \assign W' T$ \nllabel{line:mont-mul-1}\;
$Q \assign R_0 J \bmod \beta$ \nllabel{line:mont-mul-2}\;
$H \assign \lfloor Qp/\beta \rfloor$ \nllabel{line:mont-mul-3}\;
$Y' \assign R_1 - H + p$ \nllabel{line:mont-mul-4}\;
\medskip
\KwRet{$X', Y'$}\;
\caption{Modified Montgomery butterfly}
\end{algorithm}

\begin{thm}
\label{thm:montgomery}
Algorithm \ref{algo:montgomery} is correct.
\end{thm}
\begin{proof}
Lines \ref{line:mont-add-1}--\ref{line:mont-sub} are identical to the corresponding lines in Algorithm \ref{algo:new}. Line \ref{line:mont-mul-1} computes the product $W'T$, placing the low and high words of the result in $R_0$ and $R_1$ respectively, so that $0 \leq R_1 < p$. Line \ref{line:mont-mul-2} computes $Q = R_0/p \pmod \beta$ (this may be regarded as a $2$-adic approximation to the quotient $W'T/p$). Line \ref{line:mont-mul-3} computes the high word of $Qp$, so that $0 \leq H < p$. We have $Qp = R_0 \pmod \beta$, so $Qp = H\beta + R_0$. Then $W'T - Qp = \beta(R_1 - H)$, and thus $WT = R_1 - H \pmod p$. This agrees modulo $p$ with the value computed for $Y'$ in line \ref{line:mont-mul-4}, which lies in the interval $[0, 2p)$.

(Usually in Montgomery's algorithm, a further comparison and conditional subtraction/addition is performed to normalise the result into $[0, p)$, but we have simply skipped that.)
\end{proof}

\section{Implementation and performance}

The practical benefit (if any) derived from the new algorithms depends heavily on the hardware used. To illustrate what may occur in practice, we performed some timing experiments on several machines: a 3.06GHz Intel Xeon (model X5675, `Westmere' microarchitecture), a 2.6GHz Intel Xeon (model E5-2670, `Sandy Bridge' microarchitecture), and an AMD Vishera (model FX-8350, `Piledriver' microarchitecture), all modern 64-bit processors.

We compared a C implementation of a number-theoretic transform incorporating Algorithm \ref{algo:new}, a similar implementation of a transform based on Algorithm \ref{algo:ntl}, and the FFT routine in NTL itself (version 5.5.2). The code is available from the author's web site. It is moderately optimised: the inner loops are 2-way unrolled, and the last two layers of the FFT have dedicated loops. The values of $W$ and $W'$ are precomputed and stored in a table. Everything was compiled using GCC 4.6.3 on the Intel machines, and GCC 4.7.2 on the AMD machine, with the \texttt{-O2} optimisation flag. We also tried the \texttt{-O3} flag; this made no significant difference to the results. Access to the high word of the product of two 64-bit integers is obtained using the compiler's built-in support for 128-bit integer types; in NTL this is handled via inline assembly macros.

We ran transforms of length $2^{11} = 2048$. This is long enough to avoid too much loop overhead, but short enough that all memory accesses hit the L1 cache, so we may ignore locality problems. For our code we used a 62-bit prime $p$, and for NTL we used a 50-bit prime. NTL supports moduli of only up to 50 bits on a 64-bit machine, for historical reasons related to floating-point arithmetic; however, only integer arithmetic is used in the actual FFT routine.

The NTL FFT routine also performs a bit-reversal of the input array and computes a table of roots of unity on each FFT invocation. To make a fair comparison with our code, which does not perform these steps, we removed their contribution in the timing data shown below.

\begin{table}[h]
\caption{Cycles per butterfly for several implementations}
\begin{tabular}{lrrr}
\toprule
                         & Westmere & Sandy Bridge & Piledriver \\
\midrule
NTL                      & 12.0     & 9.1          &  16.7      \\
Algorithm \ref{algo:ntl} & 10.1     & 9.1          &  13.9      \\
Algorithm \ref{algo:new} & 6.9      & 5.9          &  12.0      \\
\bottomrule
\end{tabular}
\label{tab:basecase}
\end{table}

The table shows the total FFT time (CPU clock cycles) divided by the number of butterflies performed, which is $11 \times 2^{10}$. Note that $O(L)$ of the $O(L \log L)$ butterflies have $W = 1$, and these are faster than the general butterfly with $W \neq 1$. Therefore these figures slightly underestimate the cost of each general butterfly.

Two features of the table are worth pointing out. First, our implementation of Algorithm \ref{algo:ntl} is competitive with NTL. In fact, on two of the three platforms, it is slightly faster. Second, Algorithm \ref{algo:new} outperforms Algorithm \ref{algo:ntl} by a factor of about 1.5 on the Intel machines and about 1.15 on the AMD machine.

The Westmere processor can sustain a maximum throughput of one 64-bit multiplication every 2 cycles \cite{fog}. If we assume that a butterfly requires three such multiplications, then the best we can expect is 6 cycles per butterfly. The reported figure of 6.9 cycles comes fairly close to this; in other words, the multiplier is close to saturated. For Sandy Bridge, the relevant figure is either 1 cycle or 2 cycles per multiplication, depending on whether one of the operands is in a register or fetched from L1 cache. In this case the limiting factor is more likely to be cache bandwidth. Finally, for Piledriver the throughput is 4 cycles per multiplication, so again we are close to saturating the multiplier.

\section*{Acknowledgements}

The author thanks Tommy F\"arnqvist, Torbj\"orn Granlund, Niels M\"oller, Jason Papadopolous and Paul Zimmermann for stimulating conversations on this topic, and Torbj\"orn Granlund for providing access to the AMD machine. The referees provided many helpful comments. The author was partially supported by the Australian Research Council, DECRA Grant DE120101293.

\bibliographystyle{amsplain}
\bibliography{fastntt}

\end{document}